\documentclass{article}

\usepackage[english]{babel}
\usepackage{tikz}

\usepackage{pgfplots, pgfplotstable}

\usepackage[letterpaper,top=2cm,bottom=2cm,left=3cm,right=3cm,marginparwidth=1.75cm]{geometry}

\usepackage{amsmath}
\usepackage{amsthm}
\usepackage{graphicx}
\usepackage[colorlinks=true, allcolors=blue]{hyperref}
\usepackage{caption}
\usepackage{subcaption}

\newtheorem{theorem}{Theorem}
\newtheorem{lemma}{Lemma}

\usepackage{nicematrix}
\usepackage{array, makecell}
\usepackage{floatrow}

\usepackage{indentfirst}

\usepackage{authblk}

\title{Price of Anarchy of Multi-Stage Machine Scheduling Games}
\author[1]{Ho-Lin Chen}
\author[1]{Pin-Ju Huang}
\affil[1]{Department of Electrical Enginnering, National Taiwan University}
\date{}                     %% if you don't need date to appear
\begin{document}
\maketitle              % typeset the header of the contribution

\begin{abstract}
In this paper, we extend the discussion of the price of anarchy of machine scheduling games to a multi-stage machine setting. The multi-stage setting arises naturally in manufacturing pipelines and distributed computing workflows, when each job must traverse a fixed sequence of processing stages. While the classical makespan price of anarchy of $2-\frac{1}{m}$ has been established for sequential scheduling on identical machines, the efficiency loss in multi-stage scheduling has, to the best of our knowledge, not been previously analyzed.

We assume that each task follows a greedy strategy and gets assigned to the least-loaded machine upon arrival at each stage. Notably, we observe that in multi-stage environments,  greedy behavior generally does not coincide with a subgame perfect Nash equilibrium. We continue with analyzing the equilibrium under greedy choices, since it is logical for modeling selfish agents with limited computational power, and may also model a central scheduler performing the common least-load scheduling heuristics. Under this model, we first show that in single-stage scheduling, greedy choice again yields an exact price of anarchy of $2-\frac{1}{m}$. In multi-stage scheduling, we show that the completion time from one stage to the next increases by at most two times the maximum job execution time. Using this relationship, we derived the price of anarchy of multi-stage scheduling under greedy choices to lie within $[2-\frac{1}{m_{\max}}, 3 - \frac{1}{m_{\max}}]$, where $m_{\max}$ denote the maximum number of machines in one stage.

\end{abstract}
\section{Introduction}
As technology becomes more and more widespread and incorporated into every corner of our lives, autonomous independent agents will play an increasingly important role in decision-making. Environments consisting of multiple selfish, autonomous agents may involve competition over limited resources, such as competing over computing resources or physical constraints. A real-life example is the competition among individual computers for finite internet bandwidth \cite{papadimitriou2001algorithms}.  When autonomous agents act to optimize their own interests, the resulting equilibrium may not be as efficient as a fully coordinated system. The social cost of selfish behavior is often compared with the optimal social cost, under the measure of price of anarchy.

 % As the number of jobs increases, we might add in multiple servers to handle the increased number of jobs. One important issue is how users choose which server their jobs are placed on. When each user tries to minimize the completion time of their job, 

The inefficiency introduced by the selfish behavior of independent agents has been extensively studied in single stage machine scheduling scenarios, where the jobs are competing over one single set of resources. However, to the best of our knowledge, limited research has been done regarding the cost of selfish behavior in multi-stage machine scheduling, where the jobs have a sequence of procedures and need to compete over several resources. A practical example of the multi-stage process is the car assembly process. An automobile chassis moves along an assembly line, with each station performing a specific task in fixed order, such as such as welding followed by painting and assembly in a manufacturing line. Each automobile needs to follow the same assembly order. The resulting system can be modeled as a multi-stage scheduling problem. Automobiles move through the assembly line and enter a new stage of machines executing the next procedure. In this multi-stage machine setting, a natural assumption is that all welding machines have identical processing speed, all painting machines have identical processing speed, and all engine installation machines have identical processing speed. However, the processing speed of welding machines need not be the same as that of painting machines, and the same is true for other types of machines. Therefore, in our model, we assume a varying speed $s_i$ across different stages, but machines in the same stage have identical speed. By comparing the resulting Nash equilibrium with the best possible scheduling that minimizes the maximum completion time, we aim to quantify the inefficiencies inflicted by selfish behavior.

Previous research done on machine scheduling games mainly focuses on Nash equilibrium of selfish players, where no player can improve its performance by changing their behavior. The social cost of the worst Nash equilibrium is compared with the the optimal scheduling, and the ratio between the two is quantified as the price of anarchy. In our multi-stage scheduling game, we assume each player employs a greedy strategy, choosing the machine with least load when it is released to a new stage. This corresponds to the case when players command their jobs to take the fastest machine when possible. It may also occur when each job is scheduled under a central coordinated agent, while the scheduler policy is to always assign task to the least loaded machine. As scheduling over multi-stages is NP-hard, the greedy scheduling policy is a simple, naive heuristic approach to minimizing the social cost. 

Also, each job may arrive to a stage at different times, so machines follow a first-come-first-serve rule. We then compare the social cost of the optimal scheduling with the social cost when each player employs greedy strategy each stage. Note that greedy choices done sequentially does not always result in a subgame perfect Nash equilibrium (more discussion in appendix~\ref{app:Greedy-choice-SPNE}). In the single-stage setting, greedy choices of all agents does lead to a subgame perfect Nash equilibrium, but in multi-stage scheduling games, greedy choices may result in non-Nash equilibrium configurations. Therefore, our results of the ratio between the social cost of greedy choices and the social cost of optimal scheduling can be described as the \emph{price of greediness}, different from the price of anarchy defined under Nash equilibrium. To compare to previous results, we will still use the price of anarchy to refer to the price of greediness in this paper.

\subsection{Related Work} \label{sec:Related-work}
% Introduce NE & PoA
In a selfish game setting, the goal of all players is to minimize their own cost following their selfish interests. The resulting stable configuration is called a \emph{Nash equilibrium}\cite{NashJohnF.1950EPin}. To quantify the inefficiencies caused by selfish behavior, the \emph{price of anarchy}\cite{koutsoupias1999worst, roughgarden2002bad} is introduced. The price of anarchy is defined as the ratio between the social cost of the worst Nash equilibrium and the social cost of the optimal configuration.
% The social cost arising from selfish behavior is compared with that of optimal scheduling, which is achieved by a central coordinator. The ratio between the two is quantified as the \emph{price of anarchy}\

The selfish machine scheduling game involves $n$ independent selfish agents selecting from $m$ machines to minimize its own cost. The algorithmic study of machine scheduling games first dates back to the paper by Koutsoupias and Papadimitriou \cite{koutsoupias1999worst}, introducing a basic network of a fixed source and a fixed destination with $m$ parallel links. The load is formulated as the sum of the load on each individual path. The authors showed that the routing problem with $m$ parallel links is equivalent to a scheduling problem with $m$ machines with the given job sizes. A worst-case coordination ratio (equivalent to the price of anarchy) of exactly $\frac{3}{2}$ is proven for the parallel $2$-link network. 

When independent agents decide which machines to select, there are two types of models. A \emph{pure strategy} requires agents to choose only one machine, while a \emph{mixed strategy} allows agents to assign probabilities over several machines. A pure strategy Nash equilibrium always exists in selfish machine scheduling games under the KP model (the Koutsoupias and Papadimitriou model) \cite{fotakis2009structure, vocking2007selfish}. Even though the existence of Nash equilibria is guaranteed, it is NP-hard to find the best and the worst Nash equilibrium of a scheduling game\cite{angel2007algorithmic,vocking2007selfish}.

In the paper of Hassin and Yovel \cite{hassin2015sequential}, they extended the KP model to the discussion of sequential price of anarchy on $m$ identical machines. Under the assumption that tasks can enter in any arbitrary order, they proved that the sequential price of anarchy is $2-\frac{1}{m}$. They also concluded that under sequential choice, the resulting Nash equilibrium might be highly nontrivial and counterintuitive. An agent might choose a machine other than the least loaded one during its turn to steer subsequent agents toward a more favorable equilibrium. 
% Move to intro part
% In this paper, we follow the greedy assumption, assuming that agents only have limited computational power and always resort to choosing the least loaded machine upon arrival in each stage. 

% This void combatting the super hard problem of considering sequential Nash.

% Doesn't seem like this paragraph is needed
% During the scheduling process, a machine might have control over the execution order of the tasks assigned to it, which is formally called "local sequencing policies". Common policies to consider include "Smallest processing time first"(SPT) or "Longest processing time first"(LPT). The conference proceeding by Immorlicaa, Li, Mirrokni, and Schulz \cite{immorlica2009coordination} discuss in detail how various policies affect the price of anarchy.
% In this paper, we consider the most general case, which is a randomized local strategy, where tasks on a machine can be in any arbitrary order.

Depending on the relationship of task execution times of each machine $t_{im}$ (execution time of task $i$ on machine $m$), the formulation of machines can be classified into four types \cite{immorlica2009coordination}. 
\begin{itemize}
    \item \emph{Identical machines} $(P)$: $t_{i m} = t_{i m'} = t_i$ for job $i$ and machines $m$ and $m'$. Execution time of a task is identical in all machines.
    \item \emph{Related machines} $(Q)$: $t_{i m} = \frac{j_i}{s_m}$, where $j_i$ is the size of job $i$ and $s_m$ is the processing speed of machine $m$. The processing time is the job size divided by the speed of each machine.
    \item \emph{Restricted scheduling} or \emph{Bipartite scheduling} $(B)$: A task can only be executed on a limited set of machines. $t_{im} = t_i$ if $m \in M_i$ or $\infty$ otherwise.
    \item \emph{Unrelated machines} $(R)$: $t_{im}$ is an arbitrary positive number.
\end{itemize}

Previous results on the makespan price of anarchy of machine scheduling games for different types of machines is summarized in table~\ref{tab:prev-work}.

\begin{table}[h]
    \centering
    \def\arraystretch{2}%  1 is the default, change whatever you need
    \begin{NiceTabular}{| c | c | c | c | c |}
    
        \hline 
             Machine Type & \makecell{Number of\\Machines} & \makecell{Number of\\Stages} & Type of Equilibrium & \makecell{Makespan Price of Anarchy}\\ 
        \hline \hline
        Identical & 2 & 1 & \makecell{Pure Nash Equilibrium} &
        $\frac{3}{2}$ \cite{koutsoupias1999worst} \\ \hline
        Identical & $m$ & 1 & \makecell{Subgame Perfect\\Nash Equilibrium} &
        $2-\frac{1}{m}$ \cite{hassin2015sequential} \\ \hline
        Related & $2$ & 1 & \makecell{Mixed Nash Equilibrium} &
        $\phi = \frac{\left( 1 + \sqrt{5}\right)}{2}$ \cite{koutsoupias1999worst} \\ \hline
        Related & $m$ & 1 & \makecell{Pure Nash Equilibrium} &
        \makecell{$O\left( \min \{ \frac{\log m}{\log \log m}, \log \frac{s_1}{s_m} \}\right)$ \\ assuming $s_1 \geq \cdots \geq s_m$}  \cite{czumaj2007tight} \\ \hline

        Restricted & $m$ & 1 & \makecell{Mixed Nash Equilibrium} &
        $\Theta \left( \frac{\log m}{\log \log m}\right)$ \cite{awerbuch2006tradeoffs,gairing2010computing} \\ \hline
        
        Unrelated & $2$ & 1 & \makecell{Pure Nash Equilibrium} &
        $\infty$ \cite{chen2020coordination} \\ \hline
        Unrelated & $m$ & 1 & \makecell{Subgame Perfect
        \\Nash Equilibrium} &
        $\Omega(n)$, $n$ is task number \cite{chen2020coordination} \\ \hline
        Identical & $m$ & 1 & \makecell{Greedy Choices} & $2-\frac{1}{m}$ (Theorem~\ref{thm:PoA-single-stage-lower-bound},~\ref{thm:PoA-single-stage-upper-bound}) \\ \hline
        \makecell{Identical\\in each stage} & \makecell{$m_i$\\ for stage $i$} & $k$ & \makecell{Greedy Choices} & $[2-\frac{1}{m_{\max}}, 3-\frac{1}{m_{\max}}]$ (Theorem~\ref{thm:PoA-multi-stage-lower-bound},~\ref{thm:PoA-multi-stage-upper-bound}) \\ \hline
    \end{NiceTabular}
    \caption{Summary of previous work and our work on one stage machine scheduling}
    \label{tab:prev-work}
\end{table}

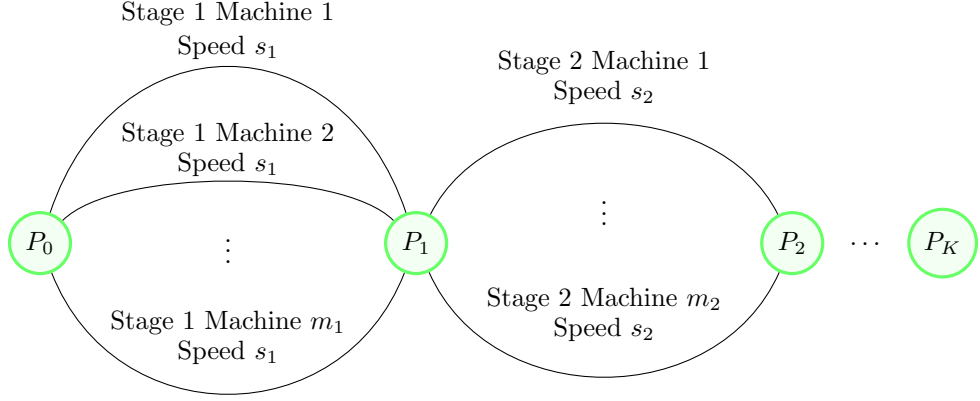
\begin{figure} [h]
    \centering
    \begin{tikzpicture}[
        roundnode/.style={circle, draw=green!60, fill=green!5, very thick, minimum size=7mm},
        squarednode/.style={rectangle, draw=red!60, fill=red!5, very thick, minimum size=5mm},
    ]
        \node [roundnode] (source) {$P_0$};
        \node [roundnode] (midpoint1) at (5,0) {$P_1$};
        \node [roundnode] (midpoint2) at (10,0) {$P_2$};
        % dots
         \node  at (11, 0) {$\dots$};
        \node [roundnode] (sink) at (12,0) {$P_K$};
        % edge 1
        \draw (source) .. controls (,3) and (4,3) .. (midpoint1);
        \node at (2.5, 3.05) {Stage $1$ Machine $1$};
        \node at (2.5, 2.6) {Speed $s_1$};

        % edge 2
        \draw (source) .. controls (1,1.0) and (4,1.0) .. 
        (midpoint1);
        \node  at (2.5, 1.45) {Stage $1$ Machine $2$};
        \node at (2.5, 1.05) {Speed $s_1$};

        % dots
         \node  at (2.5, 0) {$\vdots$};

         % edge n
        \draw (source) .. controls (1,-2.55) and (4,-2.55) .. (midpoint1);
        \node  at (2.5, -1.05) {Stage $1$ Machine $m_1$};
        \node at (2.5, -1.45) {Speed $s_1$};

        % Stage 2
        % edge 1
        \draw (midpoint1) .. controls (6,2) and (9,2) .. (midpoint2);
        \node at (7.5, 2.4) {Stage $2$ Machine $1$};
        \node at (7.5, 2) {Speed $s_2$};

        % dots
         \node  at (7.5, 0.5) {$\vdots$};

         % edge n
        \draw (midpoint1) .. controls (6,-2.25) and (9,-2.25) .. 
        (midpoint2);
        \node  at (7.5, -0.75) {Stage $2$ Machine $m_2$};
        \node at (7.5, -1.15) {Speed $s_2$};
        
    \end{tikzpicture}
    \caption{Model of our multi-stage scheduling game. All tasks get released at $P_0$ at $t=0$ in a given order. Tasks are sequentially assigned to machines in the first stage. Once execution in the first stage is finished, it enters $P_1$ and gets assigned to machine in the second stage. The tasks are greedy and is always assigned to the machine with least load when it enters a new stage. Each task repeat the process of selecting machines, queuing, and execution until it reaches the end of the $k$ stage network $P_k$. This network models a procedure where tasks of different sizes require an identical $k$ step procedure, and the goal is to minimize the maximum completion time.}
    \label{fig:Network-intro}
\end{figure}

In this paper, we propose the \emph{multi-stage machine scheduling game}, an extension to the traditional one-stage scheduling game. Instead of each job only requiring one procedure and only selecting one machine, each job now requires a series of procedures and needs to choose one machine for each procedure. It can be represented by the job passing through a series-parallel link network (See figure~\ref{fig:Network-intro}), where each link represents a machine. In the field of multi-operation scheduling, our multi-stage scheduling game models the problem of flow shop scheduling \cite{mao1995multi} as each task requires an identical sequence of $k$ procedures.

In terms of the processing sequence of jobs on a single machine, our model follows the first-come-first-serve rule. The idea comes from the dynamic flow model \cite{aronson1989survey, kotnyek2003annotated}, where the load consists of infinitesimal flow particles traversing a network and queuing up at an edge following the first come first serve rule. Our multi-stage scheduling game can be viewed as a discrete version of the dynamic flow model. Jobs that arrive earlier can select their machines and enter the queue first.

In terms of how jobs select machines in each stage, we assume that each individual job follows a greedy strategy and chooses the least loaded machine when it arrives to a new stage. At the first stage, we assume at all tasks are released at $t = 0$, In the same time, we break ties by arbitrary order for the tasks to decide which one has the priority to choose machines first. 

\subsection{Our Contributions}
Our main objective is to quantify the worst case scenario of greedy agents in a multi-stage machine scheduling setting. We first made the observation that the equilibrium under sequential greedy choice is different from the subgame perfect Nash equilibrium. However, in single-stage machine scheduling, the two definitions are identical.  Then, we compare the worst case makespan of under greedy selection of all players each stage with the makespan of the optimal scheduling. We proved an exact worst case price of anarchy of $2-\frac{1}{m}$ for the single-stage scenario. In the multi-stage scenario, we proved a relationship where the release times and completion times increases at most two times the execution time of the largest job when it goes to the next stage. Using this relationship, we were able to bound the price of anarchy of multi-stage scheduling under greedy choice between $2-\frac{1}{m_{\max}}$ and $3-\frac{1}{m_{\max}}$. \\

Our paper is organized as follows. In section~\ref{sec:preliminaries}, we introduce our new model of multi-stage machine scheduling games, and discuss how we model job processing and machine selection. Following up, in section~\ref{sec:Relation-release-completion}, we discuss the relationship between the release time and the completion time of jobs in one stage of the multi-stage scheduling game. The, we use the relationship between release time and completion time to quantify the worst case price of anarchy of single-stage and multi-stage machine scheduling games.

\section{Preliminaries} \label{sec:preliminaries}
Consider $n$ jobs waiting to be executed in $k$ stages. A job must complete its execution on stage $i$ before moving on to stage $j = i+1$. In each stage $i \in \{1, 2, \cdots k\}$, a job is only assigned to one machine to execute its job, and the job execution process is non-preemptive, meaning that the execution of any task cannot be interrupted. Each stage $i \in \{1, 2, \cdots k\}$ consists of $m_i$ identical machines, where $s_i$ is the speed of the machines. Note that different stages may have different machine speed.

% Explain Load, Nash eq, OPT
In stage $i \in \{1, 2, \cdots k\}$, the time when a job finishes completion in stage $i-1$ and enters stage $i$ is called its release time. The time it completes execution in stage $i$ is called its completion time. The release time in stage $i$ ordered from small to large is labeled $r^i_1, \ r^i_2, \cdots, r^i_n$. The job is also labeled based on its release time at each stage. Hence, the job $p^i_j$ denotes the job with release time $r^i_j$ at stage $i$. The completion time $c^i_j$ is defined as the completion time of job $p^i_j$.

The job size of job $p^i_j$ is also written as $p^i_j$. In our setting, the job execution time $t^i_j$ follows the assumption of uniform machines. 
\begin{equation}
    t^i_j = \frac{p^i_j}{s_i} 
\end{equation}

The load of the machine $\alpha$ at release time $r^i_j$, $L_{\alpha}(r^i_j)$ is defined as the speed of the machine times the completion time of the last job already assigned to machine $\alpha$ before job $p^i_j$ enters stage $i$. Notice that if many different job have same release time, $L_{\alpha}(r^i_j)$ may still be different for different jobs. Similarly, we define the load of the machine $\alpha$ at time $t$ to be the speed of the machine times the completion time of the last job already assigned to machine $\alpha$ before time $t$. At time $t$, a machine $\alpha$ in stage $i$ is \emph{idle} if its load $L_\alpha (t)$ is smaller or equal to $s_i \cdot t$; otherwise, the machine is \emph{occupied} at time $t$. If a job $p^i_j$ chooses machine $\alpha$ at stage $i$, the job $p^i_j$ starts execution in stage $i$ when it has entered stage $i$ and machine $\alpha$ is idle. Therefore, the starting time when machine $\alpha$ process job $p^i_j$ is $\max\{\frac{L_\alpha(r^i_j)}{s_i}, r^i_j\}$. The completion time of job $p^i_j$ at stage $i$ is $c^i_j = \max\{\frac{L_\alpha(r^i_j)}{s_i}, r^i_j\} + t^i_j$, and after job $p^i_j$ choose machine $m$, the load of machine $m$ also increases to $s_i \cdot \left( \max\{\frac{L_\alpha(r^i_j)}{s_i}, r^i_j\} + t^i_j \right)$.
 
At the first stage, each job is released at time $t = 0$, and chooses its machines to execute on each stage from stage $1$ to stage $k$. In the Nash equilibrium scheduling, no job can reduce its final job completion time (job completion time at stage $k$) if it deviates from its current choice of machines. The maximum final job completion time out of all the jobs is called \emph{the makespan}. In the optimal scheduling, we assume that a central scheduler arranges the job to any machines with the goal of minimizing the makespan. The ratio of the worst makespan across all Nash equilibria to the optimal makespan is defined the \emph{price of anarchy}.

\section{Price of Anarchy of Multi-Stage Greedy Machine Scheduling} \label{sec:Relation-release-completion}

In this section, we will show under our constructed recurrence relationship on release times, the bound of the completion time and the bound of the release time of each stage differs by at most a constant times of the execution time of the maximum job $t_{\max}$. Then, we will discuss in detail the upper bound and the lower bound of the price of anarchy in single-stage and multi-stage greedy machine scheduling.

\begin{figure} [h]
    \centering
    \begin{tikzpicture}[
        roundnode/.style={circle, draw=green!60, fill=green!5, very thick, minimum size=7mm},
        squarednode/.style={rectangle, draw=red!60, fill=red!5, very thick, minimum size=5mm},
    ]
        \node [roundnode] (source) {s};
        \node [roundnode] (sink) at (5,0) {t};
        % edge 1
        \draw[->] (source) .. controls (1,3) and (4,3) .. (sink);
        \node at (2.5, 2.6) {Machine $1$};

        % edge 2
        \draw[->] (source) .. controls (1,1.5) and (4,1.5) .. 
        (sink);
        \node  at (2.5, 1.45) {Machine $2$};

        % dots
         \node  at (2.5, 0) {$\vdots$};

         % edge n
        \draw[->] (source) .. controls (1,-2.55) and (4,-2.55) .. 
        (sink);
        \node  at (2.5, -1.45) {Machine $m$};
        
    \end{tikzpicture}
    \caption{An illustration of a M-machine stage, where jobs are released at source $s$ with release time $r_j$, processed through one of the machines from Machine 1 to m, then completed with completion time $c_j$}
    \label{fig:enter-label}
\end{figure}
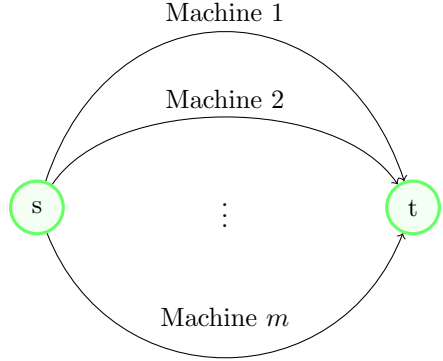

First, we will focus on the release time and completion time relationships in a single stage. Since we are focusing on a single stage, we will omit the stage index of $p^i_j$, $t^i_j$, $r^i_j$, $c^i_j$ and simply call the job size, job execution time, release time and completion time as $p_j$, $t_j$, $r_j$ and $c_j$. Jobs might switch order between stages, so we first consider the time bound of the completion time $c_j$, which is the completion time of job $p_j$, then consider how the reordering between jobs maintains the inequality. 

% Why is the formulation of the theorem 1 like this

% In the beginning of the first stage, as each job $p^1_j$ is ordered sequentially, each job is assigned to the machine with the least load at time at time $r^1_j$. The least load at time $r^1_j$ must not be greater than the average load of the machines at time $r^1_j$, which is $\frac{1}{m_1} \cdot \sum_{l=1}^{j-1} {p^1_l}$. Therefore, we have the loose bound that the start of the execution time of job ${p^1_j}$ cannot be greater than $\frac{1}{m_1 s_1} \cdot \sum_{l=1}^{j-1} j_{r^1_l}$. This provides intuition for the $\frac{1}{m s} \cdot \sum_{l=1}^{j-1} {p^1_l}$ term, which stands for the processing time when jobs are fully balanced and no machine is ever idle. 
 
 % assume that the spacing between the jobs roughly follows the same pattern as the first stage. If the first job is released at time $T$, then $p_j$ can be released at approximately $T + \frac{1}{m_1 s_1} \cdot \sum_{l=1}^{j-1} j_{r^1_l}$. Since we are considering a multi-stage machine network, the product $m_i s_i$ need not be the same at each stage. We need to relax the bound of the release time to any denominator $(ms)^* \leq m_i s_i$.  That leads to our inequality equations in lemma~\ref{thm:Induction-kTmax+sum-Unordered time}.
    
\begin{lemma} \label{thm:Induction-kTmax+sum-Unordered time}
Given a stage with $m$ machines and speed $s$, a total of $n$ jobs are passed into this stage with various release times $r_j$ ordered from first to last, satisfying $r_1 \leq r_2 \leq \dots \leq r_n$, while $p_j$ denotes the job size of the task released at time ${r_j}$. The denominator $(ms)^*$ is an arbitrary positive number that satisfies $(ms)^* \leq ms$. The size of the largest task is written as $p_{\max}$, and the execution time of the largest task is written as $t_{\max}$. If the job release time $r_j$ satisfies
\begin{equation}
    r_j \leq T + \frac{1}{(ms)^*} \sum_{l = 1}^{j-1} p_l    
    \qquad \forall j \in \{1, 2, \dots, n\}
\end{equation}

Then, the job completion time $c_j$ of each task $p_j$ satisfies

\begin{equation}
      c_j  \leq T + \left( \frac{2m-1}{ms} \right) \cdot p_{\max} + \frac{1}{(ms)^*} \sum_{l = 1}^{j-1} p_l      
      \qquad \forall j \in \{1, 2, \dots, n\}
\end{equation} \\

\begin{proof}
    We classify job executions into two cases. \\
    
    \noindent \textbf{Case I:} There is at least one machine idle when job $p_j$ enters. 
    
    The job $p_j$ experiences no delay when it enters the stage, we have
    \begin{align}
        c_j = r_j + \frac{p_j}{s} & \leq T + \frac{1}{(ms)^*} \sum_{l=1}^{j-1} p_l + \frac{p_j}{s} \\
        &\leq T  + \frac{1}{(ms)^*} \sum_{l=1}^{j-1} p_l + \left( \frac{2m-1}{ms} \right) \cdot p_{\text{max}} \label{eq:relationship-of-idle}
    \end{align}
    
    Equation~\ref{eq:relationship-of-idle} comes from the fact that $\frac{p_j}{s} \leq \frac{p_{\max}}{s} \leq \frac{2m-1}{ms} \cdot p_{\max}$. \\

    \noindent \textbf{Case II:} All machines are occupied when job $p_j$ enters.

     The job $p_j$ needs to wait for previous jobs to complete before it enters a machine. We refer back to the last job $p_f$ before $p_j$ where $p_f$ enters the stage with at least one machine idle. Just before the job $p_f$ enters the network, all machines in the system have no job waiting in the queue. If there is any job in the queue, it could have chosen the idle machine at time $r_j$. Therefore, at time $r_f$, there must be no queue in all of the machines, and one of the machine must be idle. The job size in execution in the rest $m-1$ machines must be smaller than $p_{\max}$.  We can quantify the load of machine $\alpha$ at time $r_f$ as equation~\ref{eq:load-relation-occupied-pf}.
     
    \begin{equation}
        L_m \left( r_f \right) \leq s \cdot r_f + p_{\max}\label{eq:load-relation-occupied-pf}, \quad 
        \forall \alpha \in \{1, 2, \dots ,m \}
    \end{equation}
    
    Exactly one machine is idle
    \begin{equation}
      L_\alpha \left( r_f \right) \leq s \cdot r_f,  \quad\exists \alpha \in \{1, 2, \dots ,m \}
    \end{equation}

    The sum of the load on all machines right before the job $p_j$ enters satisfies the constraint 
    \begin{equation}
        \sum_{\alpha=1}^m  L_\alpha\left( r_f \right)   \leq m s \cdot r_f + \left( m-1 \right)\cdot p_{\max}
    \end{equation}

    After jobs $p_f$ to $p_{j-1}$ enter the network, the sum of the load on all machines right before the job $p_j$ enters satisfies the constraint
    \begin{equation}
        \sum_{\alpha=1}^m  L_\alpha \left( r_j \right)  \leq m s \cdot r_f + \left( m-1 \right)\cdot p_{\max} + \sum_{l=f}^{j-1} p_l 
    \end{equation}

    Job $p_j$ chooses the machine of the least load to execute its job. The least load of the machines must be less than the average load of the machines, which is $\frac{1}{m} \sum_{\alpha=1}^m  L_\alpha(r_j)$. Therefore, we can get the upper bound of the completion time.
    \begin{align}
        c_j &\leq \frac{1}{ms} \left( m s \cdot r_f + \left( m-1 \right)\cdot p_{\max} + \sum_{l=f}^{j-1} p_l  \right) + \frac{p_j}{s} \\
        &= r_f + \frac{m-1}{ms} \cdot p_{\max} + \frac{1}{ms} \sum_{l=f}^{j-1} p_l + \frac{p_j}{s} \label{eq:completion-time-plug-in-2}\\
        &\leq T + \frac{1}{(ms)^*} \sum_{l = 1}^{f-1} p_l + \frac{m-1}{ms} \cdot p_{\max} + \frac{1}{(ms)^*} \sum_{l=f}^{j-1} j_{r_{i}} + \frac{p_{\max}}{s} \label{eq:completion-time-plug-in-3}\\
        &\leq T + \left( \frac{2m-1}{ms} \right) \cdot p_{\max} + \frac{1}{(ms)^*} \sum_{l=1}^{j-1} p_l \label{eq:completion-time-plug-in-4}
    \end{align}
    
    From Eq~\ref{eq:completion-time-plug-in-2} to Eq~\ref{eq:completion-time-plug-in-3}, we use the bound on release time $r_f \leq T + \frac{1}{(ms)^*s} \sum_{l=1}^{f-1} p_l$ and the relation that $\frac{1}{ms} \leq \frac{1}{(ms)^*}$. From Eq~\ref{eq:completion-time-plug-in-3} to Eq~\ref{eq:completion-time-plug-in-4}, we merge the terms $\frac{1}{(ms)^*} \sum_{l=1}^{f-1} p_l$  and $\frac{1}{(ms)^*} \sum_{l=f}^{j-1} p_l$ and merge the terms with the factor $p_{\max}$.

    In both case I and case II, the completion time $c_j$ follows
    \begin{align}
        c_j \leq T + \left( \frac{2m-1}{ms} \right) \cdot p_{\max} + \frac{1}{(ms)^*} \sum_{l=1}^{j-1} p_l 
    \end{align}
\end{proof}
\end{lemma}

Note that the jobs that arrive earlier will not always complete earlier. That is to say, $c_1 \leq c_2 \leq \dots \leq c_n$ is not always true. However, in order for the completion time relation in lemma~\ref{thm:Induction-kTmax+sum-Unordered time} to work, it must obey the relation $r_1 \leq r_2 \leq \dots \leq r_n$. As a result, we need to maintain the release time to be well ordered, and hence we will start by reordering the completion time in order. We consider a bijection $\sigma(j)$ for this particular stage where $c_{\sigma(j)}$ corresponds to the $\text{j}^\text{th}$ smallest completion time in $\{c_1, c_2, \dots, c_n\}$. If multiple completion times are identical, they can be arranged in any order. In multi-stage scheduling game, the permutation $\sigma_i(j)$ is written with a subscript notating its corresponding stage, but since we are only focusing on one single stage, the subscript $i$ is omitted. In lemma~\ref{thm:Reorder-Thm}, we consider the properties of the reordered completion time $c_{\sigma(j)}$.

\begin{lemma}\label{thm:Reorder-Thm}

If the job completion time $c_j$ of each task $p_j$ satisfies

\begin{equation} 
      c_j \leq T + \frac{1}{(ms)^*} \sum_{l = 1}^{j-1} p_l 
      \qquad \forall j \in \{1, 2, \dots, n\}
\end{equation} \\

Then, if we reorder the indices from $\{1, 2, \dots, n\}$ to $\{\sigma(1), \sigma(2), \dots, \sigma(n)\}$ so that job $c_{\sigma(j)}$ is the $\text{j}^\text{th}$ smallest job completion time, following $c_{\sigma(1)} \leq c_{\sigma^(2)} \leq \dots \leq c_{\sigma(n)}$. Then the tasks renumbered by indices $\{\sigma(1), \sigma(2), \dots, \sigma(n)\}$ should satisfy the new relationship
\begin{equation}
      c_{\sigma(j)} \leq T + \frac{1}{(ms)^*} \sum_{l = 1}^{j-1} p_{\sigma(l)}
      \qquad \forall j \in \{1, 2, \dots, n\}
\end{equation} \\
\end{lemma}
\begin{proof}
    The sequence $c_{\sigma(j)}$ follows $c_{\sigma(1)} \leq c_{\sigma(2)} \leq \dots \leq c_{\sigma(n)}$. Consider the task $p_j$. Define $\beta (j)$ be the minimum index of tasks with equal or greater completion time than $c_j$.

    \begin{align}
        \beta (j) \equiv& \min l  \ s.t. \ c_l \geq c_j \label{eq:bet-j-def}\\ 
        =& \min l  \ s.t. \ \sigma^{-1}(l) \geq \sigma^{-1}(j) \label{eq:beta-j-sigma-j-relation}\\ 
        =&\textstyle \min_{x = \sigma^{-1}(j)}^{n} \sigma(x) \label{eq:bet-j-sigma-x-def}
    \end{align}

    From Eq~\ref{eq:bet-j-def} to Eq~\ref{eq:beta-j-sigma-j-relation}, we use the inverse relation $\sigma^{-1}(l)$ where the job $p_j$ is the job with the $\sigma^{-1}(l)^\text{th}$ smallest completion time. For job $p_l$ to finish not later than $p_j$, its corresponding index $\sigma^{-1}(l)$ in the ordered completion time must not be greater than $\sigma^{-1}(j)$. From Eq~\ref{eq:beta-j-sigma-j-relation} to Eq~\ref{eq:bet-j-sigma-x-def}, we use the fact that $\sigma^{-1}(l)$ is an integer between $\{1, 2, \dots, n\}$, and use the substitution $l=\sigma(x)$. \\

    In other words,
    \begin{align}
        \forall l \leq \left( \beta(j) - 1\right) & \qquad c_l < c_j, \ \text{and thus} \ \sigma^{-1}(l) < \sigma^{-1}(j) \label{eq:compare-sigma-l-sigma-j} 
    \end{align}

    If $\beta(j)= 1$, in other words, the completion time $c_1$ is equal or greater than $c_j$, then 
    \begin{align}
        c_j \leq c_1 \leq T
    \end{align}
    \begin{align}
        c_j = c_{\sigma(\sigma^{-1}(j))}  \leq T \leq T + \frac{1}{(ms)^*} \sum_{l = 1}^{\sigma^{-1}(j)-1} p_{\sigma(l)} \label{eq:beta(j)-1-c_j}
    \end{align}

    If $\beta(j) \geq 2$, which is to say the completion time $c_{\beta(j)}$ is equal to or greater than $c_j$, then 
    \begin{align}
        c_j = c_{\sigma(\sigma^{-1}(j))}\ \leq c_{\beta(j)} \ \leq T + \frac{1}{(ms)^*} \sum_{l=1}^{\beta(j)-1} p_l \label{eq:c_j_p_j_beta_j_bound}
    \end{align}
    
    The set $A = \{1, 2, \dots, \beta(j)-1\}$ is a subset of the set $B = \{\sigma(1), \sigma(2), \dots, \sigma\left(\sigma^{-1}(j)-1 \right)\}$. This is because for any element $a \in A$,
    \begin{align}
        & a \leq \left( \beta(j)-1 \right) \label{eq:a<beta(j)} \\
        \Rightarrow{} & \  \sigma^{-1}(a) < \sigma^{-1}(j) \label{eq:sigma(a)<sigma(j)} \\
        \Rightarrow{} & \  a \in \{ \sigma(1), \sigma(2), \dots, \sigma(\sigma^{-1}(j) - 1) \} \label{eq:a-range}\\
        \Rightarrow{} & \ a  \in B \label{eq:a-in-B}
    \end{align}

    From Eq~\ref{eq:a<beta(j)} to Eq~\ref{eq:sigma(a)<sigma(j)}, we used the property in Eq~\ref{eq:compare-sigma-l-sigma-j}. From Eq~\ref{eq:sigma(a)<sigma(j)} to Eq~\ref{eq:a-range}, we use the fact that $\sigma^{-1}(a)$ is an integer from $1$ to $\left( \sigma^{-1}(j)-1 \right)$, and the relation $a=\sigma\left(\sigma^{-1}(a)\right)$.

    Thus, $A$ is a subset of $B$, so $\sum_{l=1}^{\beta(j)-1} p_l \leq \sum_{l=1}^{\sigma^{-1}(j)-1} p_{\sigma(l)}$. By Eq~\ref{eq:c_j_p_j_beta_j_bound}, we have
    
    \begin{align}
        c_j= c_{\sigma^{-1}(\sigma(j))} \ &\leq T + \frac{1}{(ms)^*} \sum_{l=1}^{\beta(j)-1}  p_l \\ \ &\leq T + \frac{1}{(ms)^*} \sum_{l=1}^{\sigma^{-1}(j)-1} p_{\sigma(l)} \label{eq:case-beta(j)>=2}
    \end{align}

    In both cases (Eq~\ref{eq:beta(j)-1-c_j} and Eq~\ref{eq:case-beta(j)>=2}), the result is
    \begin{align}
        c_{\sigma(\sigma^{-1}(j))} \ \leq T + \frac{1}{(ms)^*} \sum_{l=1}^{\sigma^{-1}(j)-1} p_{\sigma(l)} \label{eq:c-sigmainv(sigma)}
    \end{align}

    Let $r=\sigma^{-1}(j)$, we plug it into Eq~\ref{eq:c-sigmainv(sigma)},
    \begin{align}
        c_{\sigma(r)} \ \leq T + \frac{1}{(ms)^*} \sum_{l=1}^{r-1} p_{\sigma(l)} \label{eq:c-sigma(sigma)-final}
    \end{align}
    
    % Let the task with smallest index completed at the same time or after \textbf{task $i$} be \textbf{task $m+1$}.
    
    Therefore, the sequence $c_{\sigma^{-1}(j)}$ after ordering from small to large satisfies 
    \begin{equation}
      c_{\sigma(j)} \leq T + \frac{1}{(ms)^*} \sum_{l = 1}^{j-1} p_{\sigma(l)}   
      \qquad \forall j \in \{1, 2, \dots, n\}
    \end{equation} 
\end{proof}

Following on, we will derive the relation between the release time and the ordered completion time in theorem~\ref{thm:completion-time-one-stage}.

\begin{theorem}
\label{thm:completion-time-one-stage}
    Given a stage with $m$ machines with speed $s$, and a positive number $(ms)^* \leq ms$, a total of $n$ jobs are passed into this stage with various release times $r_j$ ordered from first to last, if the job release time $r_j$ satisfies
    \begin{equation}
          r_j \leq T + \frac{1}{(ms)^*} \sum_{l = 1}^{j-1} p_l 
          \qquad \forall j \in \{1, 2, \dots, n\}
    \end{equation}
    
    Then, for the completion time $c_j$ we construct a transformation $\sigma(j)$ such that $c_{\sigma(1)} \leq c_{\sigma(2)} \leq \dots \leq c_{\sigma(n)}$. The completion time $c_{\sigma(j)}$ satisfies
    
    \begin{equation}
          p_{\sigma(j)} \leq T + \left( \frac{2m-1}{ms} \right) \cdot p_{\text{max}} + \frac{1}{(ms)^*} \sum_{l = 1}^{n-1} p_{\sigma(l)}     \qquad \forall j \in \{1, 2, \dots, n\}
    \end{equation} \\
\end{theorem}
\begin{proof}
    Under the release time criterion  
    \begin{equation}
          r_j \leq T + \frac{1}{(ms)^*} \sum_{l = 1}^{j-1} p_l 
          \qquad \forall j \in \{1, 2, \dots, n\}
    \end{equation}
    
    Lemma~\ref{thm:Induction-kTmax+sum-Unordered time} gives us the relation of the job completion time before reordering

    \begin{equation}
        c_j \leq T + \left( \frac{2m-1}{ms} \right) \cdot p_{\max} + \frac{1}{(ms)^*} \sum_{l = 1}^{j-1} p_l  
        \qquad \forall j \in \{1, 2, \dots, n\}
    \end{equation} \\

    We reorder the index so that it follows $c_{\sigma(1)} \leq c_{\sigma(2)} \leq \dots \leq c_{\sigma(n)}$. By using lemma~\ref{thm:Reorder-Thm} to derive the relation of completion times after reordering the completion times from small to large, and substituting the $T$ in lemma~\ref{thm:Reorder-Thm} by $T + \frac{2m-1}{ms} \cdot p_{\max}$, we obtain
    
    \begin{equation}
        c_{\sigma(j)} \leq T + \left( \frac{2m-1}{ms} \right) \cdot p_{\max} + \frac{1}{(ms)^*} \sum_{l = 1}^{j-1} p_{\sigma(l)}    
        \qquad \forall j \in \{1, 2, \dots, n\}
    \end{equation} \\
\end{proof}

We have shown than the bound of the completion time increases at most by $\left(\frac{2m-1}{ms}\right) p_{\max}$ compared to the release time when the jobs go through one single stage. Following on, we will apply our results to the sequential machine scheduling game. In a multi-stage machine scheduling game, each job is ordered based on their release time, then gets assigned to the least loaded machine each stage. A job is released in the next stage whenever it reaches completion in the previous stage. We will prove upper bounds and lower bounds for single-stage and multi-stage machine scheduling games.

First, we will consider the single-stage machine scheduling game with $m$ machines and $s$ speed. The worst case price of anarchy of a single stage network is exactly $2-\frac{1}{m}$, and we will prove it by separately proving its lower bound (theorem~\ref{thm:PoA-single-stage-lower-bound}) and its upper bound (theorem~\ref{thm:PoA-single-stage-upper-bound}). 
Then, we will consider a multi-stage machine scheduling game with $k$ stages. Stage $i$ consists of $m_i$ machines with speed $s_i$. The network is illustrated in figure~\ref{fig:Identical-machine-scheduling-per-stage-identical}. We will provide a lower bound (theorem~\ref{thm:PoA-multi-stage-lower-bound}) and an upper bound (theorem~\ref{thm:PoA-multi-stage-upper-bound}) for the price of anarchy of multi-stage machine scheduling.

\begin{figure} [h]
    \centering
    \begin{tikzpicture}[
        roundnode/.style={circle, draw=green!60, fill=green!5, very thick, minimum size=7mm},
        squarednode/.style={rectangle, draw=red!60, fill=red!5, very thick, minimum size=5mm},
    ]
        \node [roundnode] (source) {$P_0$};
        \node [roundnode] (midpoint1) at (5,0) {$P_1$};
        \node [roundnode] (midpoint2) at (10,0) {$P_2$};
        % dots
         \node  at (11, 0) {$\dots$};
        \node [roundnode] (sink) at (12,0) {$P_K$};
        % edge 1
        \draw (source) .. controls (,3) and (4,3) .. (midpoint1);
        \node at (2.5, 3.05) {Stage $1$ Machine $1$};
        \node at (2.5, 2.6) {Speed $s_1$};

        % edge 2
        \draw (source) .. controls (1,1.0) and (4,1.0) .. 
        (midpoint1);
        \node  at (2.5, 1.45) {Stage $1$ Machine $2$};
        \node at (2.5, 1.05) {Speed $s_1$};

        % dots
         \node  at (2.5, 0) {$\vdots$};

         % edge n
        \draw (source) .. controls (1,-2.55) and (4,-2.55) .. (midpoint1);
        \node  at (2.5, -1.05) {Stage $1$ Machine $m_1$};
        \node at (2.5, -1.45) {Speed $s_1$};

        % Stage 2
        % edge 1
        \draw (midpoint1) .. controls (6,2) and (9,2) .. (midpoint2);
        \node at (7.5, 2.4) {Stage $2$ Machine $1$};
        \node at (7.5, 2) {Speed $s_2$};

        % dots
         \node  at (7.5, 0.5) {$\vdots$};

         % edge n
        \draw (midpoint1) .. controls (6,-2.25) and (9,-2.25) .. 
        (midpoint2);
        \node  at (7.5, -0.75) {Stage $2$ Machine $m_2$};
        \node at (7.5, -1.15) {Speed $s_2$};
        
    \end{tikzpicture}
    \caption{Illustration of a network of Multi-Stage Identical Machine Scheduling. The network consists of $k$ stage. Each stage is a parallel of $m_i$ machines with identical speed $s_i$.}
    \label{fig:Identical-machine-scheduling-per-stage-identical}
\end{figure}
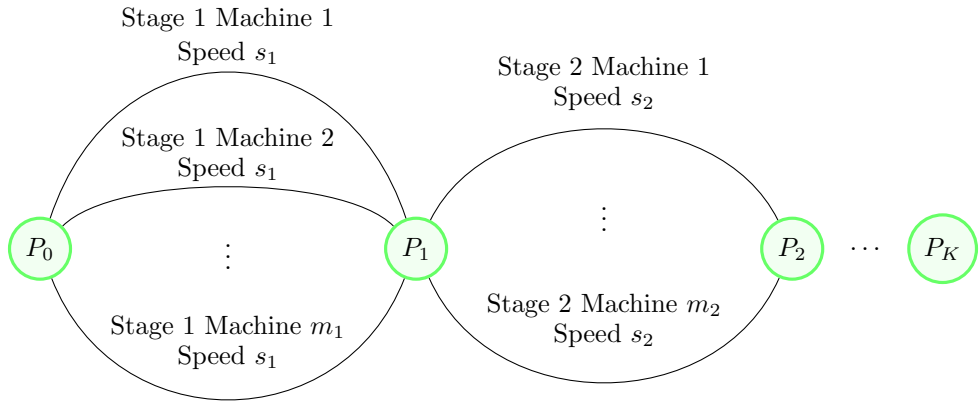

Now, let's first consider the lower bound of the price of anarchy of single stage machine scheduling.

\begin{theorem} \label{thm:PoA-single-stage-lower-bound}
 Consider a single-stage machine scheduling game with $m$ machines of speed $s$. There are $n$ jobs waiting to be processed, each with size $p_j$ and released at time $t = 0$ at the source node. Each job requires processing at every stage, but can choose either of the machines to execute. The worst case makespan price of anarchy is at least $2 - \frac{1}{m}$.
\begin{equation}
    \text{PoA} \geq 2 - \frac{1}{m}
\end{equation}
\end{theorem}

\begin{proof}
    An example of the lower bound  price of anarchy constructed by the following job distribution. The jobs to execute are $m \cdot \left(m -1 \right)$ jobs of size $1$ and $1$ job of size $m$. 
    
    In the optimal scheduling, $m \cdot \left(m -1 \right)$ jobs of size $1$ are equally distributed to the first $m - 1$ machines, and the job of size $m$ is executed on the last machine. Each machine has a load of $m$, and the optimal makespan is $\frac{m}{s}$. 
    
    In the worst case Nash equilibrium, the $m \cdot \left(m -1 \right)$ jobs of size $1$ evenly distribute between the $m$ machines first, adding a load of $m - 1$ to each individual machine. Then, the the job of size $m$ chooses a machine, causing the final makespan to be $\frac{2m-1}{s}$. In this example, the price of anarchy is 
    \begin{align}
        \text{PoA in example} = \frac{\frac{2m-1}{s}}{\frac{m}{s}} = 2 - \frac{1}{m}
    \end{align}

    The worst cast price of anarchy is at least the price of anarchy of this example, so we have
    \begin{align}
        \text{PoA} \geq 2 - \frac{1}{m}
    \end{align}
\end{proof}
% Notation:
% \begin{itemize}
%     \item $r_{i,j}$: in stage $i$, the release time of the $j^\text{th}$ job (ranked from the earliest to the latest).
%     \item in stage $i$, the completion time of the $j^\text{th}$ job (ranked from the earliest to the latest).
% \end{itemize}

We have proven that the lower bound of the price of anarchy of single-stage scheduling is at least $2-\frac{1}{m}$, established through the single stage network. Next, we will provide an upper bound of the price of anarchy in theorem~\ref{thm:PoA-single-stage-upper-bound}. Our technique essentially comes from the worst case performance of the list scheduling procedure proposed by Graham\cite{GrahamR.L.1966Bfcm}. We will show that our upper bound of price of anarchy is equivalent to the worst case performance guarantee of list scheduling in the proof of theorem~\ref{thm:PoA-single-stage-upper-bound}.

\begin{theorem} \label{thm:PoA-single-stage-upper-bound}
 Consider a single-stage machine scheduling game with $m$ machines of speed $s$. There are $n$ jobs waiting to be processed, each with size $p_j$ and released at time $t = 0$ at the source node. Each job requires processing at every stage, but can choose either of the machines to execute. The worst case makespan price of anarchy is at most $2 - \frac{1}m$. 
\begin{equation}
    \text{PoA} \leq 2 - \frac{1}{m}
\end{equation}
\end{theorem}
\begin{proof}
    Graham's list scheduling algorithm adopts a similar setting of scheduling $n$ tasks on $m$ identical parallel machines. The goal is to find the scheduling which minimizes the makespan $c_{\max}$. In Graham's list scheduling algorithm, jobs are ordered in a preference list and whenever a machine becomes idle, the first uncompleted job from the list gets assigned to the machine. They proved that under any given order of job preference list, this list scheduling algorithm has a worst-case performance guarantee of $2-\frac{1}{m}$ compared to the optimal solution.

    In our multi-stage machine scheduling game, instead of machines choosing jobs from a preference list, jobs are assigned to machine based on a given order. If we use the given job order as the preference list in list scheduling, the resulting assignment would be identical. Jobs being assigned to least loaded machine in the given order would be equivalent to idle machines always choosing the next job on the list. Thus, the worst-case guarantee of $2-\frac{1}{m}$ in list scheduling corresponds to the upper bound of the price of anarchy of $2-\frac{1}{m}$ in single-stage greedy machine scheduling.
\end{proof}

From thm~\ref{thm:PoA-single-stage-lower-bound} and thm~\ref{thm:PoA-single-stage-upper-bound}, we achieved a matching lower bound and upper bound of $2-\frac{1}{m}$ of the price of anarchy of single stage machine scheduling. Continuing on, we will analyze the multi-stage scheduling system, and achieve an upper bound on multi-stage scheduling games.

\begin{theorem} \label{thm:PoA-multi-stage-lower-bound}
 A multi-stage scheduling game consists of $k$ stages, with $m_i$ machines of speed $s_i$ in stage $i$. There are $n$ jobs waiting to be processed, each with size $p_j$ and released at time $t = 0$ at the source node. Each job requires processing at every stage, but can choose either of the machines in a stage to execute. The worst case makespan price of anarchy is at least $2 - \frac{1}{m_{\max}}$, where $m_{\max}$ is the maximum count of machines in a single stage.
 \begin{equation}
    \text{PoA} \geq 2 - \frac{1}{m_{\max}}
\end{equation}
\end{theorem}
\begin{proof}
    Let the stage with $m=m_{\max}$ be the stage $i^*$. We construct an example where the speed of the machines in each stage is assigned as follows.
    \begin{equation}
        s_i = \begin{cases}
            1 \ \ \quad \text{if} \  s = s^* \\
            \infty \quad \text{if} \  s \neq s^* \\
        \end{cases}
    \end{equation}
\end{proof}

The job distribution is $m_{\max} \cdot (m_{\max} -1)$ jobs of size 1 and a job of size $m_{\max}$. Since all stages except stage $i^*$ has infinity processing speed and zero processing time, processing time in stage $i^*$ contributes to all of the processing time.  The processing time of this network is equivalent to the processing time of a single stage network with $m=m_{\max}$. Hence, by theorem~\ref{thm:PoA-single-stage-lower-bound}, the lower bound price of anarchy of this example is $2-\frac{1}{m_{\max}}$. This example also gives a lower bound on the price of anarchy of multistage machine scheduling.
\begin{align}
    \text{PoA} \geq 2-\frac{1}{m_{\max}}
\end{align}

After we give a lower bound of price of anarchy in theorem~\ref{thm:PoA-multi-stage-lower-bound}, we will provide a proof on the upper bound of price of anarchy in multi-stage machine scheduling in theorem~\ref{thm:PoA-multi-stage-upper-bound}.

\begin{theorem} \label{thm:PoA-multi-stage-upper-bound}
 A multi-stage scheduling game consists of $k$ stages, with $m_i$ machines of speed $s_i$ in stage $i$. There are $n$ jobs waiting to be processed, each with size $p_j$ and released at time $t = 0$ at the source node. Each job requires processing at every stage, but can choose either of the machines in a stage to execute. The worst case makespan price of anarchy is at most $3 - \frac{1}{m_{\max}}$, where $m_{\max}$ is the maximum count of machines in a single stage.
\begin{equation}
    \text{PoA} \leq 3 - \frac{1}{m_{\max}}
\end{equation}
\end{theorem}
\begin{proof}
    Recall that notation $p^{i}_{j}$ denotes the job with the $j^\text{th}$ earliest release time in stage $i$, and $r^{i}_{j}$, $c^{i}_{j}$ denotes the release time and completion time of the job $p^{i}_{j}$. Let $\left(ms\right)^* = \min_{i=1}^k {m_i s_i}$. In the first stage, the release time follows the criterion 
    
    \begin{equation}
        r^{1}_{j} = 0 \leq \frac{1}{\left(ms\right)^*} \sum_{l = 1}^{j-1} p^{1}_{l} \qquad \forall j \in \{1, 2, \dots, n\}    
    \end{equation} \\

    Let $\sigma_i(j)$ denote the bijection function which reorders $c^i_j$ into the relation $c^i_{\sigma_i(1)} \leq c^i_{\sigma_i(2)} \leq \cdots \leq c^i_{\sigma_i(n)}$. By theorem~\ref{thm:completion-time-one-stage} and plugging $T$ as $0$, the completion time of stage 1 satisfies
    \begin{equation}
        c^{1}_{\sigma_1(j)} \leq  \left( \frac{2m_1-1}{m_1 s_1} \right) \cdot p_{\text{max}} + \frac{1}{\left(ms\right)^*} \sum_{l = 1}^{j-1} p^{1}_{\sigma_1(l)}  \qquad \forall j \in \{1, 2, \dots, n\}
    \end{equation} \\
    
    When a job is completed in the first stage, it goes on to the second stage. The completion time of a stage is the release time of the next stage. In addition, since $c_{\sigma_1(j)}$ follows the order $c_{\sigma_1(1)} \leq c_{\sigma_1(2)} \leq \dots \leq c_{\sigma_1(n)}$, alongside the requirement $r^2_1 \leq r^2_2 \leq \dots \leq r^2_n$, we know that $r^2_j$ corresponds to $c^1_{\sigma_1(j)}$, and $p^2_j$ corresponds to $p^1_{\sigma_1(j)}$. Therefore, in stage 2, the release time satisfies
    \begin{equation}
        r^{2}_{j} = c^1_{\sigma_1(j)} \leq  \left( \frac{2m_1-1}{m_1 s_1} \right) \cdot p_{\text{max}} + \frac{1}{\left(ms\right)^*} \sum_{l = 1}^{j-1} p^{2}_{l}  
        \qquad \forall j \in \{1, 2, \dots, n\}
    \end{equation} \\

    Again, by theorem~\ref{thm:completion-time-one-stage}, the completion time of stage 2 satisfies
    \begin{equation}
        c^{2}_{\sigma_2(j)} \leq  \left( \frac{2m_1-1}{m_1 s_1} + \frac{2m_2-1}{m_2 s_2} \right) \cdot p_{\text{max}} + \frac{1}{\left(ms\right)^*} \sum_{l = 1}^{j-1} p^{2}_{\sigma_2(l)}      
        \qquad \forall j \in \{1, 2, \dots, n\}
    \end{equation} \\

    By induction over each stage, we can obtain the completion time at the last stage, stage $k$.
    \begin{equation}
          c^{k}_{\sigma_k(j)} \leq  \sum_{i=1}^k \left( \frac{2m_i-1}{m_i s_i} \right) \cdot p_{\text{max}} + \frac{1}{\left(ms\right)^*} \sum_{l = 1}^{j-1} p^{k}_{\sigma_k(l)} 
          \qquad \forall j \in \{1, 2, \dots, n\}
    \end{equation} \\
    
    We achieve a upper bound of the makespan in the equilibrium under greedy choices.
    \begin{align} \label{eq:Equi-speed-EQU-makespan-bound}
        T_\text{equ} & \leq \sum_{i=1}^k \cdot \left( \frac{2m_i-1}{m_i s_i} \right) \cdot p_{\text{max}} + \frac{1}{\left(ms\right)^*} \sum_{l = 1}^{n-1} p^{k}_{\sigma_k(l)}
    \end{align}

    The makespan of the optimal scheduling also has two lower bounds. The first bound is required for the largest job to finish execution in all $k$ stages.
    \begin{equation} \label{eq:Equi-speed-OPT-KTmax-bound}
        T_\text{opt} \geq \sum_{i=1}^k \frac{p_\text{max}}{s_i}
    \end{equation}

    The second bound is achieved by bounding on the \textbf{rate limiting stage}. The maximum rate to pass through the rate limiting stage per unit time is $\min_{i=1}^k m_i s_i = \left(ms\right)^*$. Consequently, the makespan of the optimal is at least the total job size divided by the maximum job processing rate of the rate limiting stage.
    \begin{equation} \label{eq:Equal-speed-OPT-1/M-jobsum-bound}
        T_\text{opt} \geq \frac{1}{\left(ms\right)^*}  \sum_{l = 1}^{n} p^{k}_{\sigma^{-1}_k(l)}
    \end{equation}

    By substituting Equation. \ref{eq:Equi-speed-OPT-KTmax-bound} and Equation. \ref{eq:Equal-speed-OPT-1/M-jobsum-bound} into Equation. \ref{eq:Equi-speed-EQU-makespan-bound}, we get
    \begin{align}
        T_\text{equ} &\leq \sum_{l=1}^k \cdot \left( 2-\frac{1}{m_i} \right) \cdot \frac{p_{\text{max}}}{s_i} + \frac{1}{\left(ms\right)^*} \sum_{l = 1}^{n-1} p^{k}_{\sigma_k(l)} \label{eq:T-Nash-multi-stage-final} \\
        & \leq \left( 2-\frac{1}{m_\text{max}} \right) T_\text{opt} + T_\text{opt} = \left( 3-\frac{1}{m_\text{max}} \right) T_\text{opt} \label{eq:T-Nash-multi-stage-final-with-T-OPT}
    \end{align}
    
    From Eq~\ref{eq:T-Nash-multi-stage-final} to Eq~\ref{eq:T-Nash-multi-stage-final-with-T-OPT}, we used the relation $\forall i \in \{1, 2, \dots, k\} \quad \left( 2 - \frac{1}{m_i} \right) \leq \left( 2 - \frac{1}{m_{\max}} \right)$. Finally, we reach an upper bound of the makespan price of anarchy.
    
    \begin{equation}
       \text{Makespan PoA} = \frac{T_\text{equ}}{T_\text{opt}} \leq 3-\frac{1}{m_\text{max}}
    \end{equation}
\end{proof}

Our results on the price of anarchy if single stage and multi-stage machine scheduling can be summarized in table~\ref{tab:our-results}.

\begin{table}[h]
    \centering
    \def\arraystretch{2}%  1 is the default, change whatever you need
    \begin{NiceTabular}{| c | c | c |}  
        \hline 
              & Lower bound of price of anarchy & Upper bound of price of anarchy\\ 
        \hline \hline
        $1$ stage network & $2- \frac{1}{m}$ & $2- \frac{1}{m}$ \\ \hline
        $k$ stage network & $2- \frac{1}{m}$ & $3- \frac{1}{m}$ \\ \hline
    \end{NiceTabular}
    \caption{Summary of our results of price of anarchy of machine scheduling games. In multi-stage scheduling, $m$ stands for the maximum number of machines in a single stage.}
    \label{tab:our-results}
\end{table}

\section{Conclusion and Future Work}
This paper studied the equilibrium behavior of greedy agents in multi-stage scheduling games. We assume that each task follows a greedy strategy and gets assigned to the least-loaded machine upon arrival at each stage. We observed that the greedy choice may result in a different equilibrium compared with the subgame perfect Nash equilibrium under sequential choices. In addition, we compare the makespan of multi-stage machine scheduling game under greedy choices with the optimal scheduling. In single-stage identical machine scheduling, the price of anarchy of greedy agents is exactly $2 - \frac{1}{m}$. When we expand to multi-stage scheduling games, the price of anarchy is between the range $\left[2 - \frac{1}{m_{\max}}, 3 - \frac{1}{m_{\max}} \right]$. 

Throughout this paper, we discussed the case when each player assigns job to machines under greedy choices. For the case where each player chooses machines under subgame perfect Nash equilibrium, the sequential price of anarchy is still an open problem. Furthermore, a natural extension  is to generalize the speed in a single stage to different machine types. When every stage consists of identical machine (with machine speeds in different stages allowed to be different), we show that the multi-stage price of anarchy has an upper bound of $3$, which is only a constant worse than the single-stage price of anarchy. Similar behavior in multi-stage scheduling on other machine types also remains an open question for future research. 

% References
\bibliographystyle{splncs04}
{\footnotesize % Does not work.
% \bibliography{sample}

}

\clearpage

\appendix
\section{Greedy Choice and Subgame Perfect Nash Equilibrium} \label{app:Greedy-choice-SPNE}
% the \\ insures the section title is centered below the phrase: AppendixA

The equilibrium when players sequentially select a greedy choice does not always lead to a subgame perfect Nash equilibrium. We will provide an example where the two equilibrium are different. 

\begin{figure} [h]
    \centering
    \begin{tikzpicture}[
        roundnode/.style={circle, draw=green!60, fill=green!5, very thick, minimum size=7mm},
        squarednode/.style={rectangle, draw=red!60, fill=red!5, very thick, minimum size=5mm},
    ]
        \node [roundnode] (source) {$P_0$};
        \node [roundnode] (midpoint1) at (5,0) {$P_1$};
        \node [roundnode] (midpoint2) at (10,0) {$P_2$};
        
        \node [roundnode] (sink) at (15,0) {$P_3$};
        % edge 1
        \draw (source) .. controls (0.5,0) and (4,0) .. (midpoint1);
        \node at (2.5, 0.8) {Stage $1$ Machine $1$};
        \node at (2.5, 0.35) {Speed $s_1 = 1$};

        % Stage 2
        % edge 1
        \draw (midpoint1) .. controls (6,2) and (9,2) .. (midpoint2);
        \node at (7.5, 2.3) {Stage $2$ Machine $1$};
        \node at (7.5, 1.9) {Speed $s_2 = 5$};

         % edge n
        \draw (midpoint1) .. controls (6,-2.05) and (9,-2.05) .. 
        (midpoint2);
        \node  at (7.5, -0.75) {Stage $2$ Machine $2$};
        \node at (7.5, -1.15) {Speed $s_2 = 5$};

        % edge 1
        \draw (midpoint2) .. controls (10.5,0) and (14,0) .. (sink);
        \node at (12.5, 0.8) {Stage $3$ Machine $1$};
        \node at (12.5, 0.35) {Speed $s_3 = 0.1$};
        
    \end{tikzpicture}
    \caption{A multi-stage machine game where greedy choice does not result in subgame perfect Nash equilibrium}
    \label{fig:Greed-choice-SPNE-example}
\end{figure}
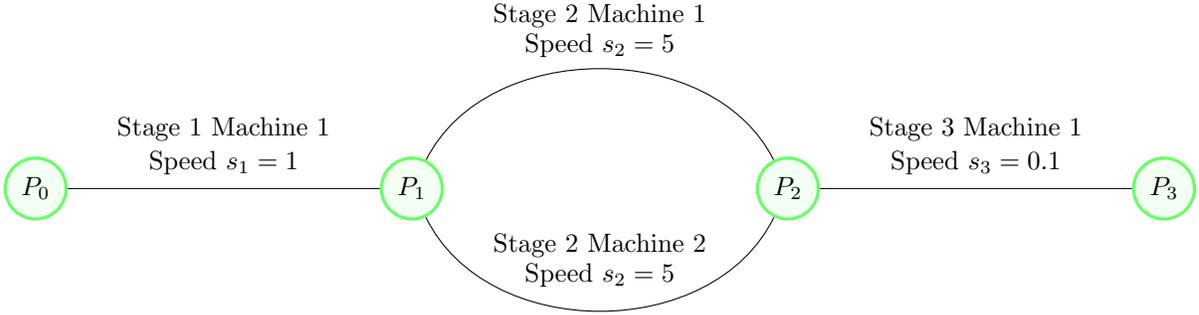

We have a multi-stage machine game with the network showw in figure~\ref{fig:Greed-choice-SPNE-example}. Suppose we have two jobs with job sizes $10$ and $1$ to process. The larger job is ordered before the smaller job. Under greedy policy, the larger job executes first in the first stage, but in the second stage, the smaller job finishes faster than the first stage. The larger job has to wait for the smaller job in the third stage. In fact, the larger job can actually improve its final completion time by letting the smaller job first in the first stage, reducing waiting time in the final stage. That is exactly what happens in the subgame perfect Nash equilibrium, and the larger job ends up having a earlier completion time. We will analyze the detail of the greedy policy and the subgame perfect Nash equilibrium in the following paragraphs.

Under greedy choice, the larger job is assigned to go first in the first stage, and completes the execution of the first stage at time $c^1_1 = 10$. The smaller job immediately begins execution at the first stage at $t=10$, and finishes execution at first stage  at time $c^1_2 = 11$. In the second stage, regardless of which machine the larger job is assigned to, the smaller job chooses the other idle machine. We have the completion time of the larger job $c^2_1 = 12$, and the completion time of the smaller job $c^2_2 = 11.2$. The smaller job gets released first in the third stage, and gets executed first in the third stage, with $c^3_1 = 21.2$. The larger job waits for the smaller job in the third stage, and ends with completion time $c^3_2 = 121.2$.

Now suppose the larger job lets the smaller job to execute first in the first stage. For the small job, the first stage completion time is $c^1_2=1$, and for the larger job, the first stage completion time is $c^1_1 = 11$. In the second stage, the smaller job is released first and finishes execution at $c^2_1 = 1.2$, while the larger job finishes at $c^2_2= 13$. The smaller job enters the third stage first and finishes execution at $c^3_1=11.2$. When the larger job is released into the third stage, the smaller job has already finished execution, so the larger job finishes at $c^3_2=113$, faster than the greed choice scenario. A summary of release times and completion times in each stage is listed in table~\ref{tab:example-of-greedy-neq-SPNE}. 

\begin{table}[h]
    \centering
    \def\arraystretch{2}%  1 is the default, change whatever you need
    \begin{NiceTabular}{| c | c | c | c | c | c | c | c |}  
        \hline 
        Policy & Task & $r^1_j$ & $c^1_j$ & $r^2_j$ & $c^2_j$ & $r^3_j$ & $c^3_j$ \\ 
        \hline \hline
        \Block{2-1}{\makecell{Greedy policy \\(Large task chooses\\ to execute first)}} 
          & Size $10$ & $r^1_1 = 0$ & $c^1_1=10$ & $r^2_1 = 10$ & $c^2_1=12$ & $r^3_2 = 12$ & $c^3_2=121.2$\\  \cline{2-8}
          & Size $1$  & $r^1_2 = 0$ & $c^1_2=11$ & $r^2_2 = 11$ & $c^2_2=11.2$ & $r^3_1 = 11.2$ & $c^3_1=21.2$\\ \hline
        \Block{2-1}{\makecell{Subgame Perfect NE \\(Large task chooses to \\ execute after smaller task)}} 
          & Size $10$ & $r^1_1 = 0$ & $c^1_1=11$ & $r^2_1 = 11$ & $c^2_2=13$ & $r^3_2 = 13$ & $c^3_2=113$\\  \cline{2-8}
          & Size $1$  & $r^1_2 = 0$ & $c^1_2=1$ & $r^2_1 = 1$ & $c^2_1=1.2$ & $r^3_1 = 1.2$ & $c^3_1=11.2$\\ \hline
    \end{NiceTabular}
    \caption{Summary of release times and completion times of greedy choice policy and the subgame perfect Nash equilibrium}
    \label{tab:example-of-greedy-neq-SPNE}
\end{table}

\end{document}